\documentclass{article}
\usepackage[utf8]{inputenc}

\usepackage[margin=1in]{geometry}

\usepackage{authblk}
\usepackage{microtype}
\usepackage{cite}
\usepackage[numbers]{natbib}
\usepackage{verbatim}
\usepackage{caption}
\usepackage{mathtools}

\usepackage{hyperref}
\hypersetup{%
	bookmarksnumbered,
	linktocpage  = true,
	colorlinks   = true,
	urlcolor     = darkblue,
	linkcolor    = darkblue,
	citecolor    = darkblue 
}
\usepackage{mleftright}

\usepackage[table]{xcolor}
\definecolor{darkblue}{rgb}{0.02, 0.17, 0.40}

\usepackage{amsfonts}
\usepackage{mathtools}
\usepackage{combelow}%

\usepackage{breakcites}

\usepackage{graphicx}
\usepackage[margin=5pt,
font+=small,labelformat=parens,labelsep=space,
skip=6pt,list=false,hypcap=false
]{subcaption}

\usepackage{csquotes}

\usepackage{amsthm}
\theoremstyle{plain}

\newtheorem{lemma}{Lemma}

\title{Encoding 3SUM}

\author[1]{Sergio Cabello}
\author[2]{Jean Cardinal}
\author[2,3]{John Iacono}
\author[2]{\\Stefan Langerman}
\author[4]{Pat Morin}
\author[2]{Aur\'{e}lien Ooms}

\affil[1]{University of Ljubljana}
\affil[2]{Universit\'e libre de Bruxelles}
\affil[3]{New York University}
\affil[4]{Carleton University}

\date{}

\begin{document}

\maketitle

\begin{abstract}
We consider the following problem: given three sets of real numbers, output
a word-RAM data structure from which we can efficiently recover the sign of the sum of any triple of numbers,
one in each set.
This is similar to a previous work by some of
the authors to encode the order type of a finite set of points. While this previous
work showed that it was possible to achieve slightly subquadratic space and
logarithmic query time, we show here that for the simpler 3SUM problem, one
can achieve an encoding that takes \(\tilde{O}(N^{\frac 32})\) space for inputs sets
of size \(N\) and allows constant time queries in the word-RAM.
\end{abstract}
\section{The Problem}

Given three sets of \(N\) real numbers
\(A = \{\, a_1 < a_2 < \cdots < a_N\,\} \),
\(B = \{\, b_1 < b_2 < \cdots < b_N\,\} \),
and \(C = \{\, c_1 < c_2 < \cdots < c_N\,\}\),
we wish to build a discrete data structure (using bits, words, and pointers) such that,
given any triple \((i,j,k) \in {[N]}^3\) it is possible to compute the sign of
\(a_i + b_j + c_k\) by only inspecting the data structure (we cannot consult
\(A\), \(B\), or \(C\)).
We refer to the map $\chi : {[N]}^3\to \{-,0,+\}, (i,j,k)\mapsto\mathrm{sgn}
(a_i+b_i+c_k)$ as the {\em 3SUM-type} of the instance $\langle A,B,C \rangle$.
Obviously, one can simply construct a lookup table of size \(O(N^3)\), such
that triple queries can be answered in \(O(1)\) time. We aim at improving on
this trivial solution.

\section{Motivation}

In the 3SUM problem, we are given an array of numbers as input and are asked
whether any three of them sum to 0. In the mid-nineties, this problem was
identified as a bottleneck of many
important problems in geometry, such as detection of affine degeneracies or
motion planning~\cite{GO95}. Since then, it has become a central problem in
fine-grained complexity theory~\cite{PW10}. It has long been conjectured to
require $\Omega (N^2)$ time. In 2014, it was shown to be solvable in $o(N^2)$
time, but no algorithm with running time $O(N^{2-\delta})$ with constant
$\delta>0$ is known~\cite{GP18}.

Lower bounds exist in restricted models of computation. Most notably,
$\Omega(N^2)$ 3-linear queries are needed to solve 3SUM~\cite{Er99},
and nontrivial lower bounds have also been proven for slightly more powerful linear
decision trees~\cite{AC05}. However, in a recent breakthrough contribution, Kane, Lovett,
and Moran showed that 3SUM could be solved using $O(N\log^2 N)$
6-linear queries~\cite{KLM18}, hence within a $O(\log N)$ factor of the
information-theoretic lower bound.

Linear decision trees are examples of {\em nonuniform algorithms}, in which we
are allowed to have different algorithms for different input sizes.
Algebraic decision trees generalize linear decision trees
by allowing decision based on the sign of constant-degree polynomials at each
node~\cite{SY82}.

Any decision tree identifying the 3SUM-type of a 3SUM instance yields a concise
encoding of this 3SUM-type:
just write down the outcome of the successive tests. Knowing the decision tree
by convention, this sequence of bits is
sufficient to recover the sign of any triple.

The question we consider here is how to make such a representation efficient,
in the sense that not only does it use merely a few bits, but the answer to any
triple query can be recovered efficiently. Understanding the interplay between
nonuniform algorithms and such data structures hopefully sheds light on the
intrinsic structure of the problem.

\section{Results}

See table~\ref{tor} for a summary. As there are only $O(N^3)$ queries, a table
of size $(\log_2 3) N^3 + O(1)$ bits suffices to give constant query time
\cite{DPT10}. This can be improved to $O(N^2\log N)$ bits of space by
storing for each pair $(i,j)$ the values
\(k_<(i,j) = \max \{ 0\}\cup \{k \colon\, a_i + b_j + c_k < 0\}\) and
\(k_>(i,j) = \min \{ N+1\}\cup \{k \colon\, a_i + b_j + c_k > 0\}\).
For a query \((i,j,k)\), we compare \(k\) against the values \(k_<(i,j)\) and \(k_>(i,j)\)
to recover \(\chi(i,j,k)\) in \(O(1)\) time. All \(k_<(i,j)\) and \(k_>(i,j)\)
can be computed in \(O(N^2)\) time via the classic quadratic time algorithm for
3SUM.

One seemingly simple representation is to store the numbers in $A$, $B$ and
$C$; however these are reals and thus we need to make them representable using
a finite number of bits.
In Section~\ref{s:numbers} we show that a minimal integer representation of a
3SUM instance may require $\Theta(N)$ bits per value, which would give
rise to a $O(N)$ query time and $O(N^2)$ space, which is far from
impressive.
In \cite{CCILO18} the problem of given a set of $N$ lines, to create an
encoding of them so that the orientation of any triple (the \emph{order type})
can be determined was studied; our problem is a special case of this where the
lines only have three slopes.
Can we do better for the case of 3SUM? We answer this in the affirmative.
In Section~\ref{s:space} we show how to use an optimal $O(N \log N)$ bits of
space with a polynomial query time. Finally, in section~\ref{s:sscqt} we show
how to use $\tilde{O}(N^{1.5})$ space to achieve $O(1)$-time queries.

\begin{table}
\centering
\caption{Table of results}\label{tor}
\begin{tabular}{cccc}
& Query time & Space (in bits) & Preprocessing time \\ \hline
Trivial & $O(1)$ & $O(N^3)$ & $O(N^3)$ \\
Almost trivial & $O(1)$ & $O(N^2 \log N)$ & $O(N^2)$ \\
Order-type encoding \cite{CCILO18} & $O(\log N)$ & $O(\frac{N^2 \log^2 \log N}{\log N})$ & $O(N^2) $\\
Order-type encoding \cite{CCILO18} & $O(\frac{\log N}{\log \log N})$ & $O(\frac{N^2 }{\log^{1-\epsilon} N})$ & $O(N^2)$ \\
Numeric representation (\S\ref{s:numbers}) & $O(N)$ & $O(N^2)$ & $N^{O(1)}$\\
Space-optimal representation (\S\ref{s:space}) & $N^{O(1)}$ & $O(N \log N)$ & $N^{O(1)}$\\
Query-optimal (\S\ref{s:sscqt}) &  $O(1)$ & $\tilde{O}(N^{1.5})$ & $O(N^{2})$ \\
\end{tabular}
\end{table}

\section{Representation by numbers} \label{s:numbers}

A first natural idea is to encode the real 3SUM instance by \emph{rounding} its numbers to integers.
We show a tight bound of $\Theta (N^2)$ bits for this representation.

\begin{lemma}
\label{lem:bitsize}
Every 3SUM instance has an equivalent integer instance
where all values have absolute value at most $2^{O(N)}$. Furthermore, there
exists an instance of 3SUM where all equivalent integer instances
require numbers at least as large as the $N$th Fibonacci number and where the
standard binary representation of the instance requires $\Omega(N^2)$ bits.
\end{lemma}

\begin{proof}
Every 3SUM instance \(A = \{\, a_1 < a_2 < \ldots < a_N\,\} \),
\(B = \{\, b_1 < b_2 < \cdots < b_N\,\} \),
and \(C = \{\, c_1 < c_2 < \cdots < c_N\,\}\)
can be interpreted as the point
\( (a_1,\ldots,a_N,b_1,\ldots,b_N,c_1,\ldots,c_N) \)
in \(\mathbb{R}^{3N}\).
Let us use the variables \(x_1,\ldots,x_N\) to encode the first \(N\) dimensions
of \(\mathbb{R}^{3N}\), \(y_1,\ldots,y_N\) to encode the next \(N\) dimensions,
and \(z_1,\ldots,z_N$ for the remaining dimensions.
Consider the subset of $\mathbb{R}^{3N}$
\[
	\Delta = \{ (x_1,\ldots,x_N,y_1,\ldots,y_N,z_1,\ldots,z_N) \mid
				x_i<x_{i+1}, ~y_j<y_{j+1}, ~ z_k<z_{k+1}~ \forall i,j,k \in [N-1]\}
\]
and the set $\Pi$ of $N^3$ hyperplanes $x_i+y_j+z_k=0$, where $i,j,k\in [N]$.
Let $\mathcal{A}$ be the arrangement defined
by $\Pi$ \emph{inside $\Delta$}. Instances of 3SUM correspond to points in $\Delta$.
Moreoever, two 3SUM instances have the same 3SUM-type if and only if they are
in the same cell of $\mathcal{A}$.

Consider an instance $\langle A,B,C \rangle$ and let $\sigma=\sigma(A,B,C)$
be the cell of $\mathcal{A}$ that contains it.
Then $\sigma$ is the cell defined by the inequalities
\begin{align*}
	\forall{i,j,k}\in [N]&:~~~
	\begin{cases}
		x_i+y_j+z_k > 0 & \text{if $\chi(i,j,k)=+1$,}\\
		x_i+y_j+z_k = 0 & \text{if $\chi(i,j,k)=0$,}\\
		x_i+y_j+z_k < 0 & \text{if $\chi(i,j,k)=-1$.}
	\end{cases}\\
	\forall{i,j,k}\in [N-1]&:~~~
		\begin{cases}
		x_i - x_{i+1}<0,\\
		y_j - y_{j+1}<0,\\
		z_k - z_{k+1}<0.
	\end{cases}
\end{align*}
Let $\sigma'$ be the subset of $\mathbb{R}^{3N}$ defined by the following inequalities:
\begin{align*}
	\forall{i,j,k}\in [N]&:~~~
	\begin{cases}
		x_i+y_j+z_k \geq 1 & \text{if $\chi(i,j,k)=+1$,}\\
		x_i+y_j+z_k = 0 & \text{if $\chi(i,j,k)=0$,}\\
		x_i+y_j+z_k \leq -1 & \text{if $\chi(i,j,k)=-1$.}
	\end{cases}\\
	\forall{i,j,k}\in [N-1]&:~~~
		\begin{cases}
		x_i - x_{i+1} \leq 1,\\
		y_j - y_{j+1} \leq 1,\\
		z_k - z_{k+1} \leq 1.
	\end{cases}
\end{align*}

Clearly $\sigma'$ is contained in $\sigma$. Moreover, for a sufficiently large $\lambda>0$
the scaled instance $\langle \lambda A,\lambda B,\lambda C \rangle$ belongs to $\sigma'$.
Therefore, $\sigma'$ is nonempty.

Since $\sigma'$ is defined by a collection of linear inequalities defining closed halfspaces,
there exists a point $p$ in $\sigma'$ defined by a subset of at most $3N$ inequalities,
where the inequalities are actually equalities. Let us assume for simplicity that
exactly $3N$ equalities define the point $p$. Then, $p=(x,y,z)$ is the solution
to a linear system of equations $M [x~ y ~z]^T=\delta$
where $M$ and $\delta$ have their entries in $\{ -1,0,1 \}$
and each row of $M$ has at most three non-zero entries. The solution $p$ to this
system of equations is an instance equivalent to $\langle \lambda A,\lambda B,\lambda C \rangle$.

Because of Cramer's rule, the system of linear equations has solution with entries
$\det(M_i)/\det(M)$,
where $M_i$ is the matrix obtained by replacing the $i$th column of $M$ by $\delta$.
We use the following simple bound on the determinant.
Since $\det(M)=\sum_{\pi}\mathrm{sgn}(\pi) \prod_i m_{i,\pi(i)}$, where
$\pi$ iterates over the permutations of $[3N]$, there are
at most $3^{3N}$ summands where $\pi$ gives non-zero product $\prod_i m_{i,\pi(i)}$ (we have
to select one non-zero entry per row), and the product is always in $\{ -1,0,1\}$.
Therefore $|\det(M)|\leq 3^{3N}$. Similarly, $|\det(M_i)|\leq 4^{3N}$ because
each row of $M_i$ has at most $4$ non-zero entries.
We conclude that the solution to the system $M [x~ y ~z]^T=\delta$
are rationals that can be expressed with $O(N)$ bits. This solution gives
a 3SUM instance with rationals that is equivalent to $\langle A, B, C \rangle$.
Since all the rationals have the common denominator ($\det(M)$), we can scale the result
by $\det(M)$ and we get an equivalent instance with integers, where
each integer has $O(N)$ bits.

The proof of the second statement is by implementing the Fibonacci recurrence in each of the
arrays $A,B,C$. This can be achieved by letting:
\begin{eqnarray*}
a_i + b_1 + c_{N-i+1} & = & 0, \text{for }i\in [N] \\
a_1 + b_i + c_{N-i+1} & = & 0, \text{for }i\in [N] \\
a_{i-1} + b_{i-2} + c_{N-i+1} & < & 0, \text{for }i\in \{3,4,\ldots ,N\},
\end{eqnarray*}
The first two sets of equations ensure that the two arrays $A$ and $B$ are identical, while
the array $C$ contains the corresponding negated numbers, in reverse order.
From the inequalities in the third group, and depending on the choice of the initial values $a_1, a_2$,
each array contains a sequence growing at least as fast as the Fibonacci sequence.
\end{proof}

Note that this is a much smaller lower bound than for order types of points sets in the plane,
the explicit representation of which can be shown to require exponentially many bits per coordinate~\cite{GPS89}.

\section{Space-optimal representation} \label{s:space}

By considering the arrangement of hyperplanes defining the 3SUM problem, we get an
information-theoretic lower bound on the number of bits in a 3SUM-type.

\begin{lemma}
There are $2^{\Theta(N\log N)}$ distinct 3SUM-types of size $N$.
\end{lemma}
\begin{proof}
3SUM-types of size $N$ are in one-to-one correspondence with cells of the
arrangement of $N^3$ hyperplanes in $\mathbb{R}^{3N}$. The
number of such cells is $O(N^{9N})$ and is easily shown to be at least
${(N!)}^2$.
\end{proof}

In order to reach this lower bound, we can simply
encode the label of the cell of the arrangement in \(\Theta(N \log N)\) bits.
However, decoding the information
requires to construct the whole arrangement which takes \(N^{O(N)}\) time.
An alternative solution is to store a
vertex of the arrangement of hyperplanes \(a_i + b_j + c_k \in \{\,
-1, 0, 1\,\}\).
There exists such a vertex that has the same 3SUM-type as the input point, as shown in the proof of Lemma~\ref{lem:bitsize}.
To answer any query, either recompute the vertex from the basis then answer the query using arithmetic,
or use linear programming.
Hence we can build a data structure of $O(N\log N)$ bits such that triple queries can be answered in polynomial time.

Note that we do not exploit much of the 3SUM structure here. In particular, the
same essentially holds for $k$-SUM, and can also be generalized to a {\sc
Subset Sum} data structure of $O(N^2)$ bits, from which we can extract the sign
of the sum of any subset of numbers.

\section{Subquadratic space and constant query time}\label{s:sscqt}
Our encoding is inspired by  Gr{\o}nlund and Pettie's $\tilde{O}(N^{1.5})$
non-uniform algorithm for 3SUM~\cite{GP18}.
Our data structure stores three components, which we call the
\emph{differences}, the \emph{staircase} and the \emph{square neighbors}.

\begin{description}
\item[Differences.] Partition $A$ and $B$ into \emph{blocks} of $\sqrt{N}$ consecutive elements. Let $D$ be the set of all differences of the form $a_i-a_j$ and $b_k-b_\ell$ where the items come from the same block. There are $O(N^{1.5})$ such differences. Sort $D$ and store a table indicating for each difference in $D$ its rank among all differences in $D$. This takes $O(\log N)$ bits for each of the $O(N^{1.5})$ differences, for a total of $O(N^{1.5}\log N)$ bits.

\item[Staircase.] Look at the table $G$ formed by all sums of the form $a_i+b_j$, which is monotonic in its rows and columns due to $A$ and $B$ being sorted and view it as being partitioned into a grid $G'$ of size $\sqrt{N}\times \sqrt{N}$ where each \emph{square} of the grid is also of size $\sqrt{N}\times \sqrt{N}$.
For each element $c \in C$, for each $i\in[1,\sqrt{N}]$ we store the largest $j$ such that some elements of the square $G'[i,j]$ are  $< c$, denote this as $V[c,i]$.
We also store, for each $c \in C$, for each $j\in[1,\sqrt{N}]$ the smallest $i$ such that some elements of the square $G'[i,j]$ are  $\geq c$, denote this as $H[c,j]$.
We thus store, in $V$ and $H$, $\sqrt{N}$ values of size $O(\log N)$ for each of the $N$ elements of $C$, for a total space usage of $O(N^{1.5}\log N)$ bits.
We call this the \emph{staircase} as this implicitly classifies, for each $c \in C$, whether each square has elements larger than $c$, smaller than $c$, or some larger and some smaller; only $O(\sqrt{N})$ can be in the last case, which we refer to as the \emph{staircase} of $c$.
\item[Square neighbors.] For each element $c \in C$, for each of the $O(\sqrt{N})$ squares on the staircase, we store the location of the predecessor and successor of $c$ in the squares $G'[i,V[c,i]]$ and $G'[H[c,j],j]$, for $i,j \in [1,\sqrt{N}]$. This takes space $O(N^{1.5}\log N)$.
\end{description}

To execute a query $(a_i,b_j,c_k)$, only a constant number of lookups in the
tables stored are needed.
If $j<\sqrt{N} \cdot H[k,i]$, then we know $a_i+b_j>c_k$.
If $i>\sqrt{N} \cdot V[k,j]$, then we know $a_i+b_j<c_k$.
If neither of these is true, then the square $G'[\lceil i/\sqrt{N}
\rceil,\lceil j/\sqrt{N} \rceil]$ is on the staircase of $c_i$ and thus using
the square neighbors table we can determine the location of the predecessor and
successor of $c_k$ in this square; suppose they are at $G[s_i,s_j]$ and
$G[p_i,p_j]$ and thus $G[s_i,s_j]\leq c_k \leq G[p_i,p_j]$. One need only
determine how these two compare to $G[i,j]=a_i+b_j$ to answer the query. But
this can be done using the differences as follows: to compare $G[s_i,s_j]$ to
$G[i,j]$ this would be determining the sign of $(a_i+b_j)-(a_{s_i}+b_{s_j})$
which is equivalent to determining the result of comparing $a_i-a_{s_i}$ and
$b_j-b_{s_j}$, which since both are in the same square, these differences are
in $D$ and the comparison can be obtained by examining their stored ranks. By
doing this for the predecessor and successor we will determine the relationship
between $a_i+b_j$ and $c_k$.

\begin{figure}
\centering
\includegraphics[trim={4cm 10cm 3cm 2cm},clip,width=5.5in]{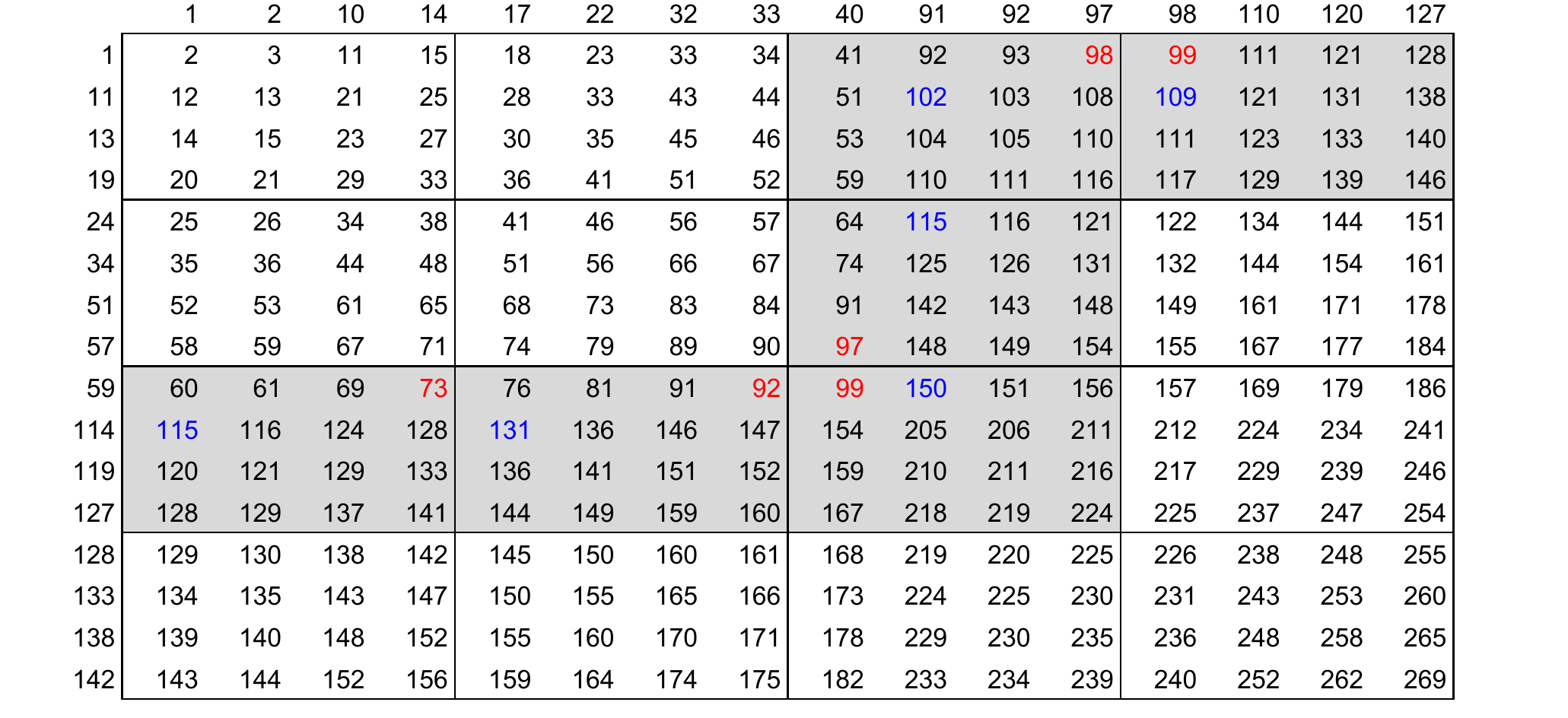}
\caption{Illustration of the staircase and square neighbors of the constant query time encoding. Here the $16\times16$ table is partitioned into a $4\times 4$ grid of squares of size $4\times 4$. If $c_k=100$, the grey illustrates the squares that form the staircase, containing values both larger and smaller than 100. Predecessors and successors within each staircase square are shown in red and blue.}
\end{figure}

\bibliography{paper}

\end{document}